\documentclass[11pt]{article}

\usepackage{amssymb,amsmath,amsthm}

\usepackage{fullpage}

\usepackage{ifpdf}
\ifpdf    
\usepackage{hyperref}
\else    
\usepackage[hypertex]{hyperref}
\fi

\usepackage{cleveref}
\usepackage{color}
\usepackage{algorithm}
\usepackage[noend]{algorithmic}
\usepackage{nicefrac}
\usepackage{aliascnt}
\usepackage{xcolor}

\usepackage{empheq}
\newcommand*\widefbox[1]{\fbox{\hspace{2em}#1\hspace{2em}}}

\numberwithin{equation}{section}

\newtheorem{theorem}{Theorem}[section]
\crefname{theorem}{Theorem}{Theorems}

\newaliascnt{lemma}{theorem}
\newtheorem{lemma}[lemma]{Lemma}
\aliascntresetthe{lemma}
\crefname{lemma}{Lemma}{Lemmas}

\newaliascnt{proposition}{theorem}
\newtheorem{proposition}[proposition]{Proposition}
\aliascntresetthe{proposition}
\crefname{proposition}{Proposition}{Propositions}

\newaliascnt{corollary}{theorem}

\aliascntresetthe{corollary}
\crefname{corollary}{Corollary}{Corollaries}

\newaliascnt{fact}{theorem}

\aliascntresetthe{fact}
\crefname{fact}{Fact}{Facts}

\newaliascnt{definition}{theorem}
\newtheorem{definition}[definition]{Definition}
\aliascntresetthe{definition}
\crefname{definition}{Definition}{Definitions}

\newaliascnt{remark}{theorem}

\aliascntresetthe{remark}
\crefname{remark}{Remark}{Remarks}

\newaliascnt{conjecture}{theorem}

\aliascntresetthe{conjecture}
\crefname{conjecture}{Conjecture}{Conjectures}

\newaliascnt{example}{theorem}

\aliascntresetthe{example}
\crefname{example}{Example}{Examples}

\newaliascnt{claim}{theorem}
\newtheorem{claim}[claim]{Claim}
\aliascntresetthe{claim}
\crefname{claim}{Claim}{Claims}

\newaliascnt{question}{theorem}

\aliascntresetthe{question}
\crefname{question}{Question}{Questions}

\newaliascnt{exercise}{theorem}

\aliascntresetthe{exercise}
\crefname{exercise}{Exercise}{Exercises}

\newaliascnt{notation}{theorem}

\aliascntresetthe{notation}
\crefname{notation}{Notation}{Notations}

\newaliascnt{problem}{theorem}

\aliascntresetthe{problem}
\crefname{problem}{Problem}{Problems}

\def\compactify{\itemsep=0pt \topsep=0pt \partopsep=0pt \parsep=0pt}

\newcommand{\abs}[1]{\left|#1\right|} 
\newcommand{\norm}[1]{\lVert#1\rVert}
\providecommand{\card}[1]{\lvert#1\rvert}

\DeclareMathOperator*{\EX}{\mathbb E}

\newcommand{\R}{\mathbb R}

\newcommand{\vect}[1]{\mathbf{#1}}
\newcommand{\minn}[1]{\min\{{#1}\}}

\providecommand{\eqdef}{:=}
\providecommand{\aset}[1]{\{ #1 \}}

\DeclareMathOperator{\fsp}{sp}

\newcommand{\cproblem}[1]{\mathsf{#1}}
\newcommand{\cclass}[1]{\mathbf{#1}}

\newcommand{\fcap}{\mathrm{cap}}
\newcommand{\fdem}{\mathrm{dem}}
\newcommand{\OPT}{\mathrm{OPT}}
\newcommand{\LP}{\mathrm{LP}}
\newcommand{\opt}{\mathrm{opt}}
\newcommand{\SDP}{\mathrm{SDP}}

\title{Cheeger-type approximation for sparsest $st$-cut%
}
\author{%
Robert Krauthgamer%
\thanks{Work supported in part by a US-Israel BSF grant \#2010418,
Israel Science Foundation grant \#897/13,
and by the Citi Foundation.
Email: \texttt{robert.krauthgamer@weizmann.ac.il, talw@mit.edu}}\\
Weizmann Institute of Science
\and
Tal Wagner\footnotemark[2]
\\
MIT
}

\begin{document}
\maketitle

\begin{abstract}
We introduce the $st$-cut version of the Sparsest-Cut problem, 
where the goal is to find a cut of minimum sparsity in a graph $G(V,E)$
among those separating two distinguished vertices $s,t\in V$.
Clearly, this problem is at least as hard as the usual (non-$st$) version.
Our main result is a polynomial-time algorithm for the product-demands setting,
that produces a cut of sparsity $O(\sqrt{\OPT})$, 
where $\OPT\leq1$ denotes the optimum
where the total edge capacity and the total demand are assumed (by normalization) 
to be $1$.

Our result generalizes the recent work of Trevisan [arXiv, 2013] for the 
non-$st$ version of the same problem (Sparsest-Cut with product demands), 
which in turn generalizes the bound achieved by the discrete Cheeger inequality,
a cornerstone of Spectral Graph Theory that has numerous applications.
Indeed, Cheeger's inequality handles graph conductance, the special case 
of product demands that are proportional to the vertex (capacitated) degrees.
Along the way, we obtain an $O(\log\card{V})$-approximation
for the general-demands setting of Sparsest $st$-Cut.
\end{abstract}

\section{Introduction}\label{section_intro}

The sparsest cut problem in graphs, defined below, 
is a fundamental optimization problem.
It is essentially equivalent to edge-expansion in graphs
and conductance in Markov chains,
and it is closely related to spectral graph theory
via a connection known as the discrete Cheeger inequality.
In terms of applications, this problem can be used as a building block 
for solving several other graph problems, 
and from a technical perspective, it is tied closely to geometric analysis, 
through the strong connection between its approximability 
to low-distortion metric embeddings. 
Given all these connections to many important problems, areas, and concepts, 
it is not surprising that sparsest cut was studied extensively.
Our focus here is on polynomial-time approximation algorithms 
for an $st$-variant of the sparsest cut problem, 
where the cut must separate two designated ``terminal'' vertices $s,t\in V$
(similarly to the minimum $st$-cut problem).

\paragraph{Sparsest cut.}

Let $G=(V,\fcap,\fdem)$ be a vertex set of size $n=|V|$
and two weight functions, called capacity and demand, 
each mapping unordered pairs of vertices to non-negative reals, 
formally, $\fcap,\fdem:\binom{V}{2}\to \R_{\ge0}$. 
It is sometimes convenient to think of $G$ as an undirected graph,
with the capacity function representing edge weights.
Denote by $C=\sum_{\aset{u,v}\subset V}\fcap(\aset{u,v})$ the total capacity,
and similarly by $D=\sum_{\aset{u,v}\subset V}\fdem(\aset{u,v})$ the total demand,
and assume both are positive. 
Define the \emph{sparsity} of the cut $(S,\bar S)$, for a subset
$\emptyset \subsetneq S \subsetneq V$,
as the ratio between the fraction of capacity separated by the cut 
and the fraction of separated demand, formally
\[ 
  \fsp_{G}(S,\bar S) 
  \eqdef \frac{\frac{1}{C}\sum_{u\in S,v\in\bar S}\fcap(\aset{u,v})}
              {\frac{1}{D}\sum_{u\in S,v\in\bar S}\fdem(\aset{u,v})} .
\]
By convention, if the denominator is zero,
e.g., in the trivial cases $S=\emptyset$ and $S=V$,
then $\fsp_{G}(S,\bar S)\eqdef\infty$.

Before proceeding, we introduce two assumptions that simplify the notation.
First, assume without loss of generality that $C=D=1$,
by simply scaling the capacities and demands.
Second, switch to a notation over \emph{ordered} pairs;
specifically, with slight abuse of notation define $\fcap: V\times V\to\R_{\ge0}$
where $\fcap(u,v)=\tfrac12 \fcap(\aset{u,v})$ for $u\neq v\in V$,
and $\fcap(v,v)=0$ for all $v\in V$;
define also $\fdem: V\times V\to\R_{\ge0}$ similarly.
Observe that under this new notation, we again have
$\sum_{u,v\in V} \fcap(u,v) = C = 1$ and
$\sum_{u,v\in V} \fdem(u,v) = D = 1$.
Overall, we obtain the more convenient form
\[ 
  \fsp_{G}(S,\bar S) 
  = \frac{2 \sum_{u\in S,v\in\bar S} \fcap(u,v)}
         {2 \sum_{u\in S,v\in\bar S} \fdem(u,v)} .
\]

In the general-demands sparsest cut problem, 
denoted henceforth $\cproblem{SparsestCut}$, 
the input is $G$ as above and the goal is to output a cut of minimum sparsity. 
An important restricted setting is that of product demands, 
where $\fdem(u,v)=\mu(u)\cdot\mu(v)$ 
for some probability distribution $\mu$ over the vertices,
and we denote this problem by $\cproblem{ProductSparsestCut}$.

\paragraph{Cheeger-type approximation.}

The well-known concept of conductance (of a graph with capacities on its edges) 
is just a special case of product demands, 
where $\mu$ is the stationary distribution of a random walk in $G$, 
that is, $\mu(v)$ is proportional to the capacitated degree of $v$, 
defined as $\mathrm{deg}(v)\eqdef \sum_{u\in V}\fcap(u,v)$. 
In this case, 
the discrete Cheeger inequality \cite{AM85, JS88, Mihail89}
efficiently produces a cut with sparsity at most $\sqrt{8\cdot \OPT}$, 
where $\OPT\leq 1$ is the sparsity of the optimal cut,%
\footnote{The normalization $C=D=1$ implies that $\OPT\leq 1$,
even in the case of general demands.
Indeed, consider the cuts $(\aset{v},V\setminus{v})$ for all $v\in V$;
the total capacity of all these cuts is $2C$,
and the total demand of all these cuts is $2D$,
hence by averaging, one of these cuts must have sparsity at most $1$.
}
see \cite{Chung97,Spielman12} for recent presentations.
This result has far-reaching theoretical implications,
e.g., for the construction of expander graphs, 
and variants of it are widely used in practice for graph partitioning tasks,
see e.g.~\cite{ShiMalik00}. 

As an extension, Trevisan \cite{Trevisan13} designed 
for more general setting of product demands,
a polynomial-time algorithm that finds a cut of sparsity $O(\sqrt{\OPT})$,
i.e., an $O(1/\sqrt{\OPT})$-factor approximation 
for $\cproblem{ProductSparsestCut}$.
His algorithm uses semidefinite programming, 
compared with a single eigenvector used in Cheeger's inequality.
Following Trevisan's terminology, 
we call such a guarantee a \emph{Cheeger-type} approximation.

\paragraph{Multiplicative approximation.}
These fundamental problems have attracted also extensive efforts 
to design polynomial-time algorithms 
with approximation factor bounded in terms of $n$.
For $\cproblem{ProductSparsestCut}$, 
a celebrated result of Arora, Rao and Vazirani \cite{ARV09}
achieves an $O(\sqrt{\log n})$-approximation.
For the more general problem $\cproblem{SparsestCut}$,
the best approximation factor known is $O(\sqrt{\log n}\log\log n)$, 
due to Arora, Lee and Naor \cite{ALN08}.
For important earlier results, see also \cite{LR99,AR98,LLR95}.

\paragraph{Results.}
We study a (new) variant of the sparsest cut problem 
concerned with cuts $(S,\bar S$) that are \emph{$st$-separating}, 
which means that $S$ contains exactly one of the vertices $s,t\in V$.
Formally, in the $st$-$\cproblem{SparsestCut}$ problem,
the input is $G=(V,\fcap,\fdem)$ as above 
together with two designated ``terminals'' $s,t\in V$, 
and the goal is to output a minimum-sparsity $st$-separating cut. 
The $st$-$\cproblem{ProductSparsestCut}$ problem is defined similarly 
in the product-demands setting. 

Our main result is an (efficient) Cheeger-type approximation 
for $st$-$\cproblem{ProductSparsestCut}$.
Along the way, 
we also obtain an $O(\log n)$-approximation for $st$-$\cproblem{SparsestCut}$.
These two results, stated formally in \cref{thm_cheeger,thm_lr},
can be viewed as extensions of \cite{Trevisan13,LLR95,AR98} 
to the $st$-setting. 
Observe that these two problems are at least as hard as their 
non-$st$ counterparts (for polynomial-time algorithms), 
because an algorithm for the former problems 
can be used to solve the latter ones with just a linear overhead, 
by fixing an arbitray $s\in V$ and trying all $t\in V$ exhaustively.

Technically, our algorithms are based on $\ell_1$-embeddings 
of certain finite metrics imposed on the vertex set,
which in turn are computed efficiently by linear and semidefinite relaxations. 
Compared to previous work, 
our distance functions have an additional property of $st$-separation, 
and our main challenge is to refine the known $\ell_1$-embedding techniques 
to ensure a separation between $s,t$.

\medskip
We additionally provide in \Cref{app:iterative} 
an $O(\sqrt{\log n})$-approximation for $st$-$\cproblem{ProductSparsestCut}$. This algorithm employs a completely different, divide-and-conquer approach,
and may be viewed as a reduction of the problem to its non-$st$ version. 
This approach does not immediately extend to a Cheeger-type approximation, because it requires an approximation factor that is a function of $n$,
and not input-dependent, as explained in \Cref{app:iterative}.

\paragraph{Related Work.}

Improved approximation bounds are known for 
$\cproblem{SparsestCut}$ and $\cproblem{ProductSparsestCut}$ 
in some special graph families, 
e.g., in bounded-treewidth graphs \cite{CKR10,GTW13,LS13}
and in planar grahps \cite{KPR93,FT03}, respectively.
See \cite{GTW13} for additional references.

On the other hand, approximating $\cproblem{SparsestCut}$ 
within a factor smaller than $17/16$ is $\cclass{NP}$-hard \cite{GTW13} 
(see \cite{MatulaS90,ChuzhoyK09,CKR10} for earlier results).
Stronger assumptions, like the unique games conjecture, 
can be used to exclude approximation within larger factors 
\cite{CKKRS06,KV05,GTW13}.
Trevisan~\cite{Trevisan13} further shows that computing a Cheeger-type 
approximation for general $\cproblem{SparsestCut}$ is Unique-Games-hard.

It is known $\cproblem{ProductSparsestCut}$ is NP-hard \cite{MatulaS90},
however all inapproximability results for this problem 
rely on stronger assumptions \cite{AMS11,RST12}.

\medskip
Apart from being a combinatorially natural problem, 
$st$-$\cproblem{SparsestCut}$ is closely related 
to popular image segmentation algorithms.
For instance, Normalized Cut \cite{ShiMalik00} is a variant of 
the graph conductance case of $\cproblem{SparsestCut}$ \cite{MajiVM11}, 
the same setting in which the discrete Cheeger inequality arises. 
For the 
application to image segmentation it is often needed to specify two predefined
points that have to be separated by the cut. This idea was used by
\cite{WuLeahy93} and later by \cite{BoykovJolly01} to reduce image
segmentation to the Minimum $st$-Cut problem, which is efficiently solvable.
However, it was noted already in \cite{WuLeahy93} that the resulting algorithm
tends to cut off isolated nodes. This motivated the introduction of normalized (or
sparse) cuts in \cite{ShiMalik00}, despite rendering the optimization
problem computationally hard. 
Followup work \cite{YuShi2004,ErikssonOK11,MajiVM11,ChewCahill15}
has attempted to encode various separation (and grouping) constraints 
into tractable relaxations of the problem, 
whose performance was then evaluated empirically. 
Our work can be viewed as a theoretical counterpart of this line of work, 
as we provide rigorous bounds for the case of $st$-separation.

\section{Basic machinery for $st$-cuts}
\label{sec:basic}
In this section we present some basic claims to reason about sparse $st$-cuts. All proofs are deferred to \Cref{sec:basic_proofs}, 
as they are simple adaptations of known arguments.

\subsection{Sparse $st$-cuts via $\ell_1$-embeddings}
\label{sec:L1embed}

We say that a cut $(S,\bar S)$ is \emph{$st$-separating}
if $S$ contains exactly one of the two vertices $s,t\in V$.
The standard approach to approximating $\cproblem{SparsestCut}$ is via embedding the vertices into $\ell_1$. 
The next lemma reproduces this argument with an additional condition 
that ensures that the produced cut is $st$-separating.

\begin{definition}
A map $f:V\rightarrow\R$ is said to be \emph{$st$-sandwiching} if $f(s) \leq f(v) \leq f(t)$ for all $v\in V$. A map $f:V\rightarrow\R^p$ is said to be \emph{$st$-sandwiching} if each of its coordinates is.
\end{definition}

\begin{lemma}\label{lmm_l1}
Let $G(V,\fcap,\fdem)$ be a $\cproblem{SparsestCut}$ instance, and let $f:V\to\R^m$. There exists a cut $(S,\bar S)$ such that
\[ \fsp_G(S,\bar S) \leq 
\frac{\sum_{u,v\in V}\fcap(u,v)\norm{f(u)-f(v)}_1}{\sum_{u,v\in V}\fdem(u,v)\norm{f(u)-f(v)}_1}, \]
and given $f$, this cut $(S,\bar S)$ is efficiently computable. 
Furthermore, if $f$ is $st$-sandwiching, then the cut is $st$-separating.
\end{lemma}

\subsection{$st$-separating semi-metrics}
\label{sec:stsep}

We now introduce semi-metrics with an additional $st$-separating property, and prove some of their useful properties. Recall that a map $d:V\times V\to\R_{\geq0}$ is called a \emph{semi-metric} if it is symmetric and satisfies the triangle inequality. The $st$-separation property we employ requires that the triangle inequality 
from $s$ to $t$ via any third point actually holds as equality.


\begin{definition} \label{defn:sep}
Let $s,t\in V$. A semi-metric $d:V\times V\to\R_{\geq0}$ is \emph{$st$-separating} if 
\[
   \forall v\in V, \qquad  d(s,t)=d(s,v)+d(v,t). 
\]
\end{definition}

As the next lemma shows, this property immediately implies that the pair $s,t$ 
attains the diameter of $V$, i.e., the maximum distance between any two points.

\begin{proposition} \label{prop:stsep}
Let $d$ be an $st$-separating semi-metric on $V$. 
Then $s,t\in V$ attain the diameter of $V$, i.e., 
\[
  \forall u,v\in V, \qquad d(u,v) \leq d(s,t).
\]
\end{proposition}

\subsection{Fr\'{e}chet embeddings}
A useful way to embed a general distance function into $\R$, 
called a Fr\'{e}chet embedding, 
is to map each point to its distance from some fixed subset $A\subseteq V$.
This simple notion is an important ingredient in many algorithms for 
$\cproblem{SparsestCut}$, including \cite{LLR95,AR98,ARV09,Trevisan13}.

\begin{definition}[Distance to a subset]
\label{fdef}
Let $d$ be a semi-metric on $V$, and let $A$ be a non-empty subset of $V$. 
The distance between a point $v\in V$ and $A$ is defined as 
$d(v,A) \eqdef \min_{a\in A}d(v,a)$.
\end{definition}

The next lemma is well-known and straightforward; its proof is omitted.
\begin{lemma}[Triangle inequality]
\label{lmm_triangle}
For every $u,v\in V$ and $A\subseteq V$,
$ d(v,A)\leq d(v,u)+d(u,A) $.
\end{lemma}

To preserve the $st$-separation property, 
we introduce the following variants of a Fr\'{e}chet embedding. 
They will be used in \Cref{sec:lr} to obtain an $O(\log n)$-approximation
(similarly to \cite{LLR95}), 
and then in the ``easy'' case of a Cheeger-type approximation 
in \Cref{sec:Cheeger} (similarly to \cite{ARV09} and \cite{Trevisan13}).

\begin{definition}
\label{def_fsigma}
Let $d$ be an $st$-separating semi-metric on $V$, 
and let $A$ be a non-empty subset of $V$. 
For each sign $\sigma\in\aset{\pm 1}$, let $f_A^\sigma:V\to\R$ be given by 
\[ 
  f_{d,A}^\sigma(v) = \tfrac{1}{2}\left[ d(v,s)+\sigma \cdot d(v,A) \right].
\]
When the metric $d$ is clear from the context, we omit it from the subscript and denote $f_A^\sigma(v)$. Define also the shorthands $f_A^+ \eqdef f_A^{+1}$ and $f_A^- \eqdef f_A^{-1}$. Lastly, define $f^\pm_A:V\rightarrow\R^2$ as $f^\pm_A:=(f^+_A,f^-_A)$.
\end{definition}

The latter map has the following key properties.

\begin{proposition}[$2$-Lipschitzness]\label{prop_lip}
For every $u,v\in V$,
$ \norm{f_A^\pm(u)-f_A^\pm(v)}_1 \leq 2 \cdot d(u,v) $.
\end{proposition}

\begin{proposition}\label{prop_goodsign}
For every $u,v\in V$,
$ 
  \norm{f_A^\pm(u)-f_A^\pm(v)}_1
  \geq \frac{1}{2}\left|d(u,A)-d(v,A)\right|
$.
\end{proposition}

\begin{proposition}\label{prop_separation}
$f^\pm_A$ is $st$-sandwiching.
\end{proposition}

\section{$O(\log n)$-approximation for general demands}\label{sec:lr}

In this section we prove the following theorem.

\begin{theorem}\label{thm_lr}
There is a randomized polynomial-time algorithm that given an instance $G$ of $st$-$\cproblem{SparsestCut}$ with $n$ vertices, outputs a cut of sparsity at most $O(\log n)\cdot \OPT$, 
where $\OPT$ is the optimal sparsity of an $st$-separating cut in $G$.
\end{theorem}

\paragraph{LP relaxation of $st$-$\cproblem{SparsestCut}$.}
Given an instance $G=(V,\fcap,\fdem)$ of $st$-$\cproblem{SparsestCut}$, denote by $\chi_S$ the characteristic function of an (arbitrary) optimal cut $(S,\bar S)$. 
The map $d_S(u,v)=\left|\chi_S(u)-\chi_S(v)\right|$ is a semi-metric on $V$, 
and thus $\cproblem{SparsestCut}$ can be relaxed to an LP that optimizes over all semi-metrics $d$ (see \cite{LR99,LLR95,AR98}).
In the $st$-$\cproblem{SparsestCut}$ case, 
the same $d_S$ is furthermore $st$-separating (Definition \ref{defn:sep}).
As usual, the objective is to minimize the ratio 
$\frac{\sum_{u,v\in V}\fcap(u,v)\cdot d(u,v)}{\sum_{u,v\in V}\fdem(u,v)\cdot d(u,v)}$, 
and by scaling the semi-metric we can assume the denominator equals $1$,
while maintaining the $st$-separating property.
We have thus proved the next lemma.

\begin{lemma} \label{lem:LPrelax}
LP \eqref{lp} is a relaxation of $st$-$\cproblem{SparsestCut}$.
\end{lemma}

\begin{empheq}[box=\widefbox]{align*}\tag{P1}\label{lp} 
& \min && \sum_{u,v\in V}\fcap(u,v)\cdot d(u,v) && \\
& \mathrm{s.t.} && \sum_{u,v\in V}\fdem(u,v)\cdot d(u,v) = 1 && \\
&  && d(v,v) = 0 && \forall v\in V \\
&  && d(u,v) \geq 0 && \forall u,v\in V \\
&  && d(u,v) = d(v,u) && \forall u,v\in V \\
&  && d(u,v) \leq d(u,w)+d(w,v) && \forall u,v,w\in V \\ 
&  && d(s,t) = d(s,v)+d(v,t) && \forall v\in V 
\end{empheq}

For the rounding procedure we use the following theorem by Bourgain \cite{Bourgain85} and Linial, London and Rabinovich \cite{LLR95}.

\begin{theorem}\label{thm_bourgain}
Let $d$ be a semi-metric on $V$, with $|V|=n$. 
There are subsets $A_1,\ldots,A_p\subseteq V$ for $p=O(\log^2n)$,
such that 
\begin{equation}\label{eq_bourgain}
  \forall u,v\in V, \qquad
  \frac{1}{O(\log n)}\cdot d(u,v) 
    \leq \frac{1}{p}\sum_{q=1}^p\left| d(u,A_q)-d(v,A_q)\right| 
    \leq d(u,v) .
\end{equation}
Moreover, the sets $A_1,\ldots,A_p$ can be computed in randomized polynomial time.
\end{theorem}

\begin{proof}[Proof of \cref{thm_lr}]
Given an instance $G=(V,\fcap,\fdem)$ of $st$-$\cproblem{SparsestCut}$ with $|V|=n$ , set up and solve LP \eqref{lp}. Denote its optimum by $\LP$ and let $d:V\times V\to\R$ be a solution that attains it. Observe that $d$ is an $st$-separating semi-metric on $V$, 
and that Lemma \ref{lem:LPrelax} implies $\LP \leq \OPT$.

Apply \cref{thm_bourgain} and let $A_1,\ldots,A_p$ be the resulting subsets. For each $i=1,\ldots,p$, define the maps $f_{A_i}^+,f_{A_i}^-,f_{A_i}^\pm$ as in \cref{def_fsigma}. By \cref{prop_goodsign}, 
\[ 
  \forall u,v\in V, \qquad 
  \norm{f_{A_i}^\pm(u)-f_{A_i}^\pm(v)}_1 
  \geq \frac{1}{2}\left|d(u,A_i)-d(v,A_i)\right| . 
\]
By summing these over all $u,v\in V$ with appropriate multipliers,
\begin{equation}\label{eq_goodsign}
  \sum_{u,v\in V}\fdem(u,v)\cdot\norm{f_{A_i}^\pm(u)-f_{A_i}^\pm(v)}_1
  \geq\frac{1}{2}\sum_{u,v\in V}\fdem(u,v)\cdot|d(u,A_i)-d(v,A_i)| .
\end{equation}
Define $g:V\to\R^{2p}$ as
\[ g(v)=\tfrac{1}{2p}\left(f_{A_1}^+(v),f_{A_1}^-(v),\ldots,f_{A_p}^+(v),f_{A_p}^-(v)\right). \]
By \cref{eq_goodsign} and the first inequality in \cref{eq_bourgain}, we get
\begin{align}
  \sum_{u,v\in V}\fdem(u,v)\cdot\norm{g(u)-g(v)}_1 
  &\geq \frac{1}{4p}\sum_{i=1}^p\sum_{u,v\in V}\fdem(u,v)\cdot\left|d(u,A_i)-d(v,A_i)\right| \nonumber\\
  &\geq \frac{1}{O(\log n)}\sum_{u,v\in V}\fdem(u,v)\cdot d(u,v). 
  \label{eq_lr_lower}
\end{align}
At the same time, by \cref{prop_lip}, every $i$ and $u,v\in V$ satisfy $\norm{f_{A_i}^\pm(u)-f_{A_i}^\pm(v)}_1 \leq 2d(u,v)$, thus
\begin{equation}\label{eq_lr_upper}
\norm{g(u)-g(v)}_1 = \frac{1}{2p}\sum_{i=1}^p\norm{f_{A_i}^\pm(u)-f_{A_i}^\pm(v)}_1 \leq d(u,v) .
\end{equation}
Putting \cref{eq_lr_lower,eq_lr_upper} together,
\[
\frac{\sum_{u,v\in V}\fcap(u,v)\cdot\norm{g(u)-g(v)}_1}{\sum_{u,v\in V}\fdem(u,v)\cdot\norm{g(u)-g(v)}_1}
\leq \frac{\sum_{u,v\in V}\fcap(u,v)\cdot d(u,v)}{\sum_{u,v\in V}\fdem(u,v)\cdot d(u,v)}\cdot O(\log n)
= \LP\cdot O(\log n) \]

Consequently, applying \cref{lmm_l1} to $g$ produces a cut $(S,\bar S)$ with sparsity $sp_G(S,\bar S)\leq O(\log n)\cdot \LP\leq O(\log n)\cdot \OPT$. By \cref{prop_separation}, for each $i$ the map $f_{A_i}^\pm$ is $st$-sandwiching, hence so is $G$, and therefore \cref{lmm_l1} further asserts that $(S,\bar S)$ is an $st$-separating cut.
\end{proof}

\paragraph{Extensions.}
If the demand function is supported only inside some subset $K\subsetneq V$
(formally, $\fdem(u,v)>0$ holds only when both $u,v\in K$),
then essentially the same proof achieves approximation 
$O(\log \card{K})$, similarly to \cite{LLR95,AR98}.

If $G$ (more precisely, the graph defined by the nonzero capacities) 
excludes a fixed minor and the demands are product demands, 
then essentially the same proof achieves $O(1)$-approximation, 
similarly to \cite{KPR93,Rao99,FT03,LS13,AGGNT13}.
Such $O(1)$-approximation is also achieved by the approach described in \Cref{app:iterative}.

\section{Cheeger-type approximation for product demands}
\label{sec:Cheeger}

Recall that an instance of $st$-$\cproblem{ProductSparsestCut}$ is $G=(V,\fcap,\mu)$, 
where $\mu$ is a probability distribution over the vertex set $V$,
and the demand function is defined accordingly as $\fdem(u,v)=\mu(v)\mu(v)$.
In this section we prove the following theorem. 

\begin{theorem}\label{thm_cheeger}
There is a randomized polynomial-time algorithm that given an instance $G$ of $st$-$\cproblem{ProductSparsestCut}$ with $n$ vertices, outputs a cut with sparsity at most $O(\sqrt{\OPT})$, 
where $\OPT$ is the optimal sparsity of an $st$-separating cut in $G$.
\end{theorem}

As mentioned in \Cref{section_intro}, Trevisan \cite{Trevisan13} proved a similar result for the usual (non-$st$) version of $\cproblem{ProductSparsestCut}$. His algorithm employs a semidefinite programming relaxation proposed by Goemans and by Linial (and used in \cite{ARV09} and followup work).
This relaxation is based on the triangle inequality constraint
\[ 
  \norm{\vect x_u-\vect x_v}^2_2 
  \leq \norm{\vect x_u-\vect x_w}^2_2+\norm{\vect x_w-\vect x_v}^2_2 , 
  \qquad \forall u,v,w\in V, 
\]
which forces $d(u,v)=\norm{\vect x_u-\vect x_v}^2_2$ to be a semi-metric. As in \Cref{sec:lr}, we modify the relaxation to force this semi-metric to be $st$-separating.

\paragraph{SDP relaxation of $st$-$\cproblem{ProductSparsestCut}$.}
\begin{empheq}[box=\widefbox]{align*}\tag{P2}\label{sdp}
& \min && \sum_{u,v\in V}\fcap(u,v)\cdot\norm{\vect x_u-\vect x_v}^2_2 && \\
& s.t. && \sum_{u,v\in V}\mu(u)\mu(v)\cdot\norm{\vect x_u-\vect x_v}^2_2 = 1 && \\
&  && \norm{\vect x_u-\vect x_v}^2_2 \leq \norm{\vect x_u-\vect x_w}^2_2+\norm{\vect x_w-\vect x_v}^2_2 && \forall u,v,w\in V \\ 
&  && \norm{\vect x_s-\vect x_t}^2_2 = \norm{\vect x_s-\vect x_v}^2_2+\norm{\vect x_v-\vect x_t}^2_2 && \forall v\in V 
\end{empheq}

\begin{lemma}
SDP \eqref{sdp} is a relaxation of $st$-$\cproblem{ProductSparsestCut}$.
\end{lemma}
\begin{proof}
Given an $st$-separating cut $(S,\bar S)$,
set $\alpha \eqdef 2\sum_{u\in S,v\in\bar S} \mu(u)\mu(v)$,
and consider a one-dimensional (i.e., real-valued) solution to SDP \eqref{sdp} 
where $x_u=0$ for $u\in S$, and $x_u=\alpha^{-1/2}$ for $u\in\bar S$.
This solution can be verified to satisfy all the constraints of SDP \eqref{sdp},
and its objective value is exactly $\fsp_{G}(S,\bar S)$.
The lemma follows by letting the cut $(S,\bar S)$ be an optimal solution for
the problem.
\end{proof}

To round a solution to \eqref{sdp}, 
we consider two cases, similarly to \cite{LR99,ARV09,Trevisan13}. 
In the first case, we get a constant factor approximation 
using the tools of \Cref{sec:stsep}. 
The second case is more difficult and will require a new approach to maintain the $st$-separation.

\begin{lemma}\cite[Lemma 4]{Trevisan13} \label{lmm_cases}
Let $d$ be a semi-metric on a point set $V$, and $\mu$ a probability distribution over $V$. At least one of the following two holds:
\begin{itemize} \compactify
\item[I.] Dense ball: There is $o\in V$ such that $B=\{v\in V:d(v,o)\leq\frac{1}{4}\}$, the ball centered at $o$ with radius $\frac{1}{4}$, satisfies $\mu(B)\geq\frac{1}{2}$.
\item[II.] No dense ball: $\Pr_{u,v\sim\mu}[d(u,v)>\frac{1}{4}]\geq\frac{1}{2}$, where $u,v$ are sampled independently from $\mu$.
\end{itemize}
\end{lemma}
\begin{proof}
Suppose the first condition fails. Sample $v\sim\mu$. The ball $B_v$ centered at $v$ with radius $\frac{1}{4}$ surely satisfies $\mu(B_v)<\frac{1}{2}$, 
hence when sampling $u\sim\mu$, there is probability at least $\frac{1}{2}$ for $u$ to be at distance at least $\frac{1}{4}$ from $v$, and the second condition holds.
\end{proof}

We consider henceforth the semi-metric $d(u,v)=\norm{\vect x_u-\vect x_v}^2_2$ 
derived from a solution to SDP \eqref{sdp},
and handle the two cases of \cref{lmm_cases} separately.

\subsection{Case I: Dense ball}

\begin{lemma}\label{lmm_denseball}
Let $G=(V,\fcap,\mu)$ be an instance of $st$-$\cproblem{ProductSparsestCut}$. Denote by $\SDP$ the optimum of \eqref{sdp} and let $\{\vect x_v\}_{v\in V}$ be an optimal solution to it. Suppose there is $o\in V$ such that the ball $B=\{v\in V:\norm{\vect x_v-\vect x_o}^2_2\leq\frac{1}{4}\}$ satisfies $\mu(B)\geq\frac{1}{2}$. Then a cut with sparsity $O(\SDP)$ can be efficiently computed.
\end{lemma}
\begin{proof}
Denote $d(u,v)=\norm{\vect x_u-\vect x_v}^2_2$ and note that $d(\cdot,\cdot)$ is an $st$-separating semi-metric on $V$. Starting with the first constraint of \eqref{sdp}, we have
\begin{align*}
  1 &= \sum_{u,v\in V}\mu(u)\mu(v)\cdot d(u,v) 
    \leq \sum_{u,v\in V}\mu(u)\mu(v)\cdot\left(d(u,o)+d(o,v)\right) 
    = 2\sum_{v\in V}\mu(v)\cdot d(v,o) 
    \\
    &\leq  2\sum_{v\in V}\mu(v)\cdot\left(d(v,B)+\frac{1}{4}\right)
    = 2\sum_{v\in V}\mu(v)\cdot d(v,B)+\frac{1}{2},
\end{align*}
where the inequality in the second line is by $d(v,o)\leq d(v,v')+d(v',o)\leq d(v,B)+\tfrac{1}{4}$, with $v'$ being the closest point to $v$ in $B$. Rearranging the above, we get
\begin{equation}\label{eq_ball}
\sum_{v\in V}\mu(v)\cdot d(v,B) = \sum_{v\notin B}\mu(v)\cdot d(v,B) \geq \frac{1}{4} 
\end{equation}
and therefore,
\begin{align}
  \nonumber 
  & \sum_{u,v\in V} \mu(u)\mu(v)\cdot\left|d(u,B)-d(v,B)\right| 
  \geq \sum_{u\in B,v\notin B}\mu(u)\mu(v)\cdot\left|d(u,B)-d(v,B)\right| 
  \\
  \label{eq_denseball}
  &= \sum_{u\in B,v\notin B}\mu(u)\mu(v)\cdot d(v,B) = \mu(B)\sum_{v\notin B}\mu(v)\cdot d(v,B) 
  \geq \frac{1}{8}, 
\end{align}
where the final inequality is by plugging \cref{eq_ball} and the hypothesis $\mu(B)\geq\frac{1}{2}$.

Use $d(\cdot,\cdot)$ and $B\subseteq V$ to define the map $f_B^\pm$ as in \cref{def_fsigma}. 
Then by \cref{prop_goodsign} and then \cref{eq_denseball},
\[
  \sum_{u,v\in V}\mu(u)\mu(v)\cdot\norm{f_B^\pm(u)-f_B^\pm(v)}_1
  \geq\frac12\sum_{u,v\in V}\mu(u)\mu(v)\cdot|d(u,B)-d(v,B)| \geq \frac{1}{16} .
\]
At the same time, by \cref{prop_lip}, for every $u,v\in V$
we have $\norm{f_B^\pm(u)-f_B^\pm(v)}_1\leq 2d(u,v)$ and hence,
\[
\sum_{u,v\in V}\fcap(u,v)\cdot\norm{f_B^\pm(u)-f_B^\pm(v)}_1 \leq
2\sum_{u,v\in V}\fcap(u,v)\cdot d(u,v) = 2\cdot\SDP .
\]
Together,
\[
\frac{\sum_{u,v\in V}\fcap(u,v)\cdot\norm{|f_B^\pm(u)-f_B^\pm(v)}_1}{\sum_{u,v\in V}\mu(u)\mu(v)\cdot\norm{f_B^\pm(u)-f_B^\pm(v)}_1} \leq 32\cdot \SDP,
 \]
and thus applying \cref{lmm_l1} to $f_B^\pm$ produces a cut $(S,\bar S)$ with sparsity $sp_G(S,\bar S)\leq 32\cdot \SDP\leq 32\cdot \OPT$. 
By \cref{prop_separation} $f_B^\pm$ is $st$-sandwiching, and hence \cref{lmm_l1} further asserts that $(S,\bar S)$ is an $st$-separating cut.
\end{proof}

\subsection{Case II: No dense ball}

\begin{lemma}\label{lmm_nodenseball}
Let $G(V,\fcap,\mu)$ be an instance of $st$-$\cproblem{ProductSparsestCut}$. Denote by $\SDP$ the optimum of \ref{sdp} and let $\{\vect x_v\}_{v\in V}$ be an optimal solution to it. 
Suppose $\Pr_{u,v\sim\mu}[\norm{\vect x_u-\vect x_v}_2^2>\frac{1}{4}]\geq\frac{1}{2}$, where $u,v$ are sampled independently from $\mu$. Then a cut of sparsity $O(\sqrt{\SDP})$ can be efficiently computed.
\end{lemma}

\begin{proof}
Let $m$ denote the dimension of the $SDP$ solution $\{\vect x_v\}_{v\in V}$. By rotation and translation, we may assume without loss of generality that $\vect x_s=\vect0\in\R^m$ and that $\vect x_t$ is in the direction of $\vect e_1\in\R^m$, the first vector in the standard unit basis. 
We treat the latter direction as a ``distinguished'' one, and for each $v\in V$ we write $\vect x_v=(y_v,\vect z_v)$, where $y_v\in\R$ is the first coordinate and $\vect z_v\in\R^{m-1}$ is the vector of the remaining coordinates. 
Under this notation, we have $y_s=0$ and $\vect z_s=\vect z_t=\vect0$,
and let us denote $T \eqdef y_t\ge 0$.
The following claim records some useful facts.
\begin{claim}\label{xyz_claim}
For all $u,v\in V$,
\begin{enumerate} \compactify
\renewcommand{\theenumi}{(\alph{enumi})}
\item \label{it:xyz_a}
$ \norm{\vect x_u-\vect x_v}_2
  \leq \abs{y_u-y_v} + \norm{\vect z_u-\vect z_v}_2
  \leq \sqrt2\norm{\vect x_u-\vect x_v}_2
$;

\item \label{it:xyz_b}
$\norm{\vect z_v}_2\leq T$; and
\item \label{it:xyz_c}
$y_v\in[0,T]$.
\end{enumerate}
\end{claim}
\begin{proof}
\begin{enumerate} \renewcommand{\theenumi}{(\alph{enumi})}
\item 
By definition, $\norm{\vect x_u-\vect x_v}_2^2=\left|y_u-y_v\right|^2+\norm{\vect z_u-\vect z_v}_2^2$.
Applying now the well-known inequality 
$a^2+b^2 \leq (a+b)^2$ for $a,b\ge0$,
gives the claimed lower bound.
For the claimed upper bound, apply similarly 
$\tfrac{a^2+b^2}{2} \geq (\tfrac{a+b}{2})^2$.
\item 
The last constraint in SDP~\eqref{sdp} implies 
$\norm{\vect x_v-\vect x_s}_2^2 \leq \norm{\vect x_t-\vect x_s}_2^2$. 
Plugging $\vect x_s=\vect0$ and $\vect x_t=(T,0,\ldots,0)$, we get $\norm{\vect x_v}_2\leq T$. 
Recalling that $\vect x_v=(y_v,\vect z_v)$,
we get $y_v^2 + \norm{\vect z_v}_2^2 = \norm{\vect x_v}_2^2 \leq T^2$.
\item The above proof of \cref{it:xyz_b} also shows that $|y_v|\leq T$, so we are left to show $y_v\geq0$. 
And, using SDP~\eqref{sdp} again yields 
$\abs{y_v-T}^2 \leq \norm{\vect x_v - \vect x_t}_2^2 \le \norm{\vect x_t-\vect x_s}_2^2 = T^2$, which implies $y_v\geq0$.
\end{enumerate}
\end{proof}

\paragraph{Step 0: Random projection.}
We now turn to the main part of the proof. We embed $\{\vect x_v\}_{v\in V}$ into $\R$ as follows. Let $\vect g\in\R^{m-1}$ be a random vector of independent standard Gaussians. We define $f_{\vect g}^{(0)}:V\rightarrow\R$ as
\[f_{\vect g}^{(0)}(v) = y_v + \tfrac{1}{6}\langle \vect z_v,\vect g\rangle \]

We begin by showing that $f_{\vect g}^{(0)}(s)$ approximately preserves, in expectation, the (non-squared) $\ell_2$-distances between the points.

\begin{claim}\label{f0_claim}
For all $u,v\in V$,
\[ 
  \tfrac{1}{16}\norm{\vect x_u-\vect x_v}_2 
  \leq \EX_{\mathbf g}\left|f_{\vect g}^{(0)}(u)-f_{\vect g}^{(0)}(v)\right| 
  \leq \sqrt2 \norm{\vect x_u-\vect x_v}_2 . 
\]
\end{claim}
\begin{proof}
By rotational symmetry of the Gaussian distribution, 
$f_{\vect g}^{(0)}(u)-f_{\vect g}^{(0)}(v) = \left(y_u-y_v\right) + \tfrac{1}{6}\langle \vect z_u-\vect z_v,\vect g \rangle$
is distributed like
$\left(y_u-y_v\right) + \tfrac{1}{6}\norm{\vect z_u-\vect z_v}_2\cdot N(0,1)$,
where $N(0,1)$ is a Gaussian distribution. 
Recalling that the first absolute moment of $N(0,1)$ is $\sqrt{{2}/{\pi}}$,
we get 
\begin{align*}
  \EX_{\mathbf g}\left|f_{\vect g}^{(0)}(u)-f_{\vect g}^{(0)}(v)\right| 
  \leq \left|y_u-y_v\right| + \tfrac{1}{6}\sqrt{\tfrac{2}{\pi}}\norm{\vect z_u-\vect z_v}_2 
  \leq \sqrt2 \norm{\vect x_u-\vect x_v}_2 ,
\end{align*}
where the final inequality is by \cref{xyz_claim}\ref{it:xyz_a} (noting that $\tfrac{1}{6}\sqrt{\tfrac{2}{\pi}}<1$).
In the other direction, with probability $\frac{1}{2}$ the terms $\left(y_u-y_v\right)$ and $\tfrac{1}{6}\norm{\vect z_u-\vect z_v}_2\cdot N(0,1)$ have the same sign, thus
\begin{align*}
  \EX_{\mathbf g}\left|f_{\vect g}^{(0)}(u)-f_{\vect g}^{(0)}(v)\right| 
  \geq \tfrac{1}{2}\Big( \left|y_u-y_v\right| + \tfrac{1}{6}\norm{\vect z_u-\vect z_v}_2\cdot\EX \left|N(0,1)\right| \Big) 
  \geq \tfrac{1}{12}\sqrt{\tfrac{2}{\pi}}\norm{\vect x_u-\vect x_v}_2 ,
\end{align*}
where again the final inequality is by \cref{xyz_claim}\ref{it:xyz_a}.
\end{proof}

\cref{f0_claim} is already sufficient to obtain a cut with sparsity $O(\sqrt{\SDP})$, but it is not guaranteed to be $st$-separating.
To resolve this, we reason as follows. 
Observe that regardless of $\vect g$, 
we have $f_{\vect g}^{(0)}(s)=0$ and $f_{\vect g}^{(0)}(t)=T$,
so if all the images of $f$ were guaranteed to lie in the interval $[0,T]$, 
then $f$ is $st$-sandwiching and
we could use \cref{lmm_l1} to produce an $st$-separating cut.
However, this is not necessarily the case, 
and the remainder of this proof overcomes this issue 
by manipulating $f_{\vect g}^{(0)}$ in two steps: 
The first step ``clips'' $f_{\vect g}^{(0)}$ into a slightly bigger interval 
$[-\frac{1}{3}T,\frac{4}{3}T]$, which has additional $T/3$ margin in each side,
and the second step ``flips'' these margin areas back into $[0,T]$. 
Since these manipulations do not affect $f_{\vect g}^{(0)}(s)=0$ and $f_{\vect g}^{(0)}(t)=T$,
our challenge is to preserve the original $\ell_2$-distances.

\paragraph{Step 1: Clipping.}
We define $f_{\vect g}^{(1)}:V\rightarrow\R$ as the clipping of $f_{\vect g}^{(0)}$ into the interval $[-\frac{1}{3}T,\frac{4}{3}T]$. Formally,
\[ 
  f_{\vect g}^{(1)}(v) = 
  \begin{cases} 
    \frac{4}{3}T &\mbox{if $f_{\vect g}^{(0)}(v)>\frac{4}{3}T$;} \\
    f_{\vect g}^{(0)}(v) & \mbox{if $f_{\vect g}^{(0)}(v)\in[-\frac{1}{3}T,\frac{4}{3}T]$;} \\
    -\frac{1}{3}T & \mbox{if $f_{\vect g}^{(0)}(v)<-\frac{1}{3}T$.}
  \end{cases}
\]

\begin{claim}\label{clipping_claim_upper}
For all $u,v\in V$, $|f_{\vect g}^{(1)}(u)-f_{\vect g}^{(1)}(v)|\leq|f_{\vect g}^{(0)}(u)-f_{\vect g}^{(0)}(v)|$.
\end{claim}
\begin{proof}
It is straightforward that the clipping operation may only decrease distances.
\end{proof}

\begin{claim}\label{clipping_claim_lower}
For $u,v\in V$, define the following three events:
\begin{itemize} \compactify
\item $\mathcal{A}_1 \eqdef \aset{ f_{\vect g}^{(0)}(v)\in [-\frac{1}{3}T,\frac{4}{3}T] }$,
\item $\mathcal{A}_2 \eqdef \aset{ f_{\vect g}^{(0)}(u)\in [-\frac{1}{3}T,\frac{4}{3}T] }$,
\item $\mathcal{A}_3 \eqdef \aset{ |f_{\vect g}^{(0)}(u)-f_{\vect g}^{(0)}(v)|\geq 
\tfrac{1}{6}\norm{x_u-x_v}_2 }$.
\end{itemize}
Let $\ell_{uv}$ be an indicator random variable for their intersection.
Then 
$ \EX_{\vect g}[\ell_{uv}] 
  = \Pr[\mathcal{A}_1\cap\mathcal{A}_2\cap\mathcal{A}_3]
  \geq\frac{1}{20}
$.
\end{claim}

\begin{proof}
First, we claim that $\Pr[\mathcal{A}_1]>\frac{19}{20}$.
Indeed, $f_{\vect g}^{(0)}(v) = y_v + \tfrac{1}{6}\langle \vect z_v,\vect g\rangle$ is distributed like $y_v + \tfrac{1}{6}\norm{\vect z_v}_2g_v$ where $g_v\sim N(0,1)$. 
The Gaussian $g_v$ has probability $>\frac{19}{20}$ to be inside 
the interval $[-2,2]$. 
By \cref{xyz_claim} we have
$\left|y_v\right|\in[0,T]$ and $\norm{\vect z_v}_2\leq T$,
which imply event $\mathcal{A}_1$.

Second, we claim that $\Pr[\mathcal{A}_2]>\frac{19}{20}$.
Indeed, the argument is the same argument as for $\mathcal{A}_1$.

Third, we claim that $\Pr[\mathcal{A}_3]>\frac{3}{20}$.
Indeed, $\left|f_{\vect g}^{(0)}(u)-f_{\vect g}^{(0)}(v)\right|$ is distributed like $\left|\left(y_u-y_v\right) + \tfrac{1}{6}\norm{\vect z_u-\vect z_v}_2g_{uv} \right|$ for $g_{uv}\sim N(0,1)$. 
The Gaussian $g_{uv}$ has probability $>\frac{3}{20}$ to be at least one standard deviation away from its mean, in the direction that agrees with the sign of $y_u-y_v$. 
In that case, 
\[ 
  \left|f_{\vect g}^{(0)}(u)-f_{\vect g}^{(0)}(v)\right|\geq
  \left|y_u-y_v\right|+\tfrac{1}{6}\norm{\vect z_u-\vect z_v}_2
  \geq \tfrac{1}{6}\norm{\vect x_u-\vect x_v}_2 , 
\]
where the second inequality is by \cref{xyz_claim}\ref{it:xyz_a}. 

Finally, a union bound now implies $\EX[\ell_{uv}] = \Pr[\mathcal{A}_1\cap\mathcal{A}_2\cap\mathcal{A}_3]\geq\frac{1}{20}$. 
\end{proof}

For every $u,v\in V$,
if both events $\mathcal A_1$ and $\mathcal A_2$ occur, then the clipping operation has no effect on $u$ and $v$, i.e.~$f_{\vect g}^{(1)}(u)=f_{\vect g}^{(0)}(u)$ and $f_{\vect g}^{(1)}(v)=f_{\vect g}^{(0)}(v)$. If furthermore event $\mathcal A_3$ occurs, then we have
\[ |f_{\vect g}^{(1)}(u)-f_{\vect g}^{(1)}(v)|=|f_{\vect g}^{(0)}(u)-f_{\vect g}^{(0)}(v)| 
   \geq \tfrac{1}{6}\norm{\vect x_u-\vect x_v}_2.  \]
This implies that for all realizations of $\vect g$
(in particular without assuming whether events 
$\mathcal A_1,\mathcal A_2,\mathcal A_3$ hold or not)
\begin{equation} \label{eq_ell}
  |f_{\vect g}^{(1)}(u)-f_{\vect g}^{(1)}(v)| \geq \tfrac{1}{6}\ell_{uv}\norm{\vect x_u-\vect x_v}_2 .
\end{equation}

\paragraph{Step 1a: Fixing a function.} 
We now aim to fix a function $f_{\vect g}^{(1)}$ (i.e., a realization of $\vect g$)
and use it in the remainder of the algorithm.
We start with arguing (non-constructively) that a good realization exists,
and will later employ an additional idea to refine it into an efficient algorithm.
Using \cref{clipping_claim_upper,f0_claim}, 
we get by linearity of expectation that
\begin{equation} \label{eq_numerator_expectation}
  \EX_{\vect g}\left[ \sum_{u,v\in V} \fcap(u,v)\cdot\left|f_{\vect g}^{(1)}(u)-f_{\vect g}^{(1)}(v)\right|\right]
  \leq
  \sqrt2 \sum_{u,v\in V} \fcap(u,v) \norm{\vect x_u - \vect x_v}_2 .
\end{equation}
At the same time, using \cref{clipping_claim_lower} and \cref{eq_ell},
\begin{equation*} \label{eq_denominator_expectation}
\EX_{\vect g}\left[\sum_{u,v\in V} \mu(u)\mu(v)\cdot \left|f_{\vect g}^{(1)}(u)-f_{\vect g}^{(1)}(v)\right|\right] 
\geq
\frac{1}{120}\sum_{u,v\in V} \mu(u)\mu(v)\cdot\norm{\vect x_u-\vect x_v}_2.
\end{equation*}
Combining these two and applying an averaging argument, 
there must exist a realization of $\vect g$ such that 
\begin{equation*}\label{maineq_existence}
  \frac
  { \sum_{u,v\in V} \fcap(u,v) \cdot\left|f_{\vect g}^{(1)}(u)-f_{\vect g}^{(1)}(v)\right|}
  {\sum_{u,v\in V} \mu(u)\mu(v) \cdot\left|f_{\vect g}^{(1)}(u)-f_{\vect g}^{(1)}(v)\right|}
  \leq 120\sqrt2 \cdot 
  \frac
  { \sum_{u,v\in V} \fcap(u,v) \cdot\norm{\vect x_u - \vect x_v}_2}
  {\sum_{u,v\in V} \mu(u)\mu(v) \cdot\norm{\vect x_u - \vect x_v}_2}.
\end{equation*}

Next, we refine this analysis into an efficient method for finding 
a realization of $\vect g$ that satisfies a similar inequality. 
We will need the following simple observation.

\begin{lemma}\label{conclemma}
Let $Z$ be a random variable taking values in the range $[0,M]$,
and let $\mu \leq \EX[Z]$.
Then $\Pr[Z > \frac{1}{2}\mu ] \geq \frac{\mu}{2 M}$.
\end{lemma}
\begin{proof}
Denote $p=\Pr[Z > \frac{1}{2}\mu]$.
Then 
$
  \mu
  \leq \EX[Z]
  \leq p\cdot M + (1-p)\cdot \frac{1}{2}\mu 
  \leq p\cdot M + \frac{1}{2}\mu
$,
which yields the lemma by simple manipulation.
\end{proof}

We now apply \cref{conclemma} to the random variable 
$Z \eqdef \sum_{u,v\in V} \mu(u)\mu(v)\cdot \ell_{uv}\norm{\vect x_u-\vect x_v}_2 $.
Observe that $Z$ is always bounded by 
$M \eqdef \sum_{u,v\in V} \mu(u)\mu(v)\cdot \norm{\vect x_u-\vect x_v}_2 $,
and since by \cref{clipping_claim_lower} its expectation is 
$\EX[Z] \ge \frac{1}{20} M$,
we get 
$
  \Pr[ Z > \frac{1}{40} M ] 
  \geq \frac{1}{40}
$.
Plugging \cref{eq_ell} into the definition of $Z$, we arrive at
\[
  \Pr\left[ \sum_{u,v\in V} \mu(u)\mu(v)\cdot |f_{\vect g}^{(1)}(u)-f_{\vect g}^{(1)}(v)| 
            > \frac{1}{240} M
     \right]
  \geq 
  \frac{1}{40}.
\]
At the same time, using \cref{eq_numerator_expectation} 
and applying Markov's inequality,
\[
  \Pr\left[
     \sum_{u,v\in V} \fcap(u,v)\cdot\left|f_{\vect g}^{(1)}(u)-f_{\vect g}^{(1)}(v)\right|
     \geq
     80\sqrt2 \sum_{u,v\in V} \fcap(u,v) \norm{\vect x_u - \vect x_v}_2
  \right]
  \leq \frac{1}{80}.
\]
Putting the last two inequalities together, 
both events hold with probability at least $\frac{1}{80}$ 
(which can be amplified by independent repetitions),
in which case we find a realization of $\vect g$ satisfying
\begin{equation}\label{maineq}
  \frac
  { \sum_{u,v\in V} \fcap(u,v) \cdot\left|f_{\vect g}^{(1)}(u)-f_{\vect g}^{(1)}(v)\right|}
  {\sum_{u,v\in V} \mu(u)\mu(v) \cdot\left|f_{\vect g}^{(1)}(u)-f_{\vect g}^{(1)}(v)\right|}
  \leq 240 \cdot 80 \cdot \sqrt2 \cdot 
  \frac
  { \sum_{u,v\in V} \fcap(u,v) \cdot\norm{\vect x_u - \vect x_v}_2}
  {\sum_{u,v\in V} \mu(u)\mu(v) \cdot\norm{\vect x_u - \vect x_v}_2} .
\end{equation}
From now on we fix such $\vect g$ and the corresponding map $f_{\vect g}^{(1)}$.

\paragraph{Step 2: Flipping.}
Recall that our current function $f_{\vect g}^{(1)}$ is confined to the interval $[-\frac{1}{3}T,\frac{4}{3}T]$.
In order to confine it to $[0,T]$, 
we eliminate the margin intervals $[-\frac{1}{3}T,0]$ and $[T,\frac{4}{3}T]$ 
by ``flipping'' (or rather, ``reflecting'') them into the main interval $[0,T]$, 
while also ``shrinking'' them by an appropriate factor. 
Formally, for $\alpha\in[0,1]$, define $f^{(2)}_\alpha:V\rightarrow\R$ by
\[ 
  f^{(2)}_\alpha(v) = 
  \begin{cases} 
    T-\alpha\left(f_{\vect g}^{(1)}(v)-T\right) &\mbox{if $f_{\vect g}^{(1)}(v)>T$}; \\
    f_{\vect g}^{(1)}(v) & \mbox{if $f_{\vect g}^{(1)}(v)\in[0,T]$}; \\
    -\alpha\cdot f_{\vect g}^{(1)}(v) & \mbox{if $f_{\vect g}^{(1)}(v)<0$}.
  \end{cases}  
\]

\begin{claim}\label{flipping_claim1}
Let $u,v\in V$. For all $\alpha\in[0,1]$ 
we have $|f^{(2)}_\alpha(u)-f^{(2)}_\alpha(v)|\leq|f_{\vect g}^{(1)}(u)-f_{\vect g}^{(1)}(v)|$.
\end{claim}

\begin{proof}
Observe that the transition from $f_{\vect g}^{(1)}$ to $f^{(2)}_\alpha$ may only decrease distances.
\end{proof}

\begin{claim}\label{flipping_claim2}
Let $u,v\in V$, and consider a uniformly random $\alpha\in\{\tfrac{1}{3},1\}$. 
Then
\[
  \EX_{\alpha\in\aset{1/3,1}} \abs{f^{(2)}_\alpha(u)-f^{(2)}_\alpha(v)} 
  \geq \tfrac16 \abs{f_{\vect g}^{(1)}(u)-f_{\vect g}^{(1)}(v)}.
\]
\end{claim}
\begin{proof}
Suppose without loss of generality that 
$f_{\vect g}^{(1)}(v) < f_{\vect g}^{(1)}(u)$. 
Consider separately the following cases:

\begin{itemize} \compactify

\item Both $f_{\vect g}^{(1)}(v),f_{\vect g}^{(1)}(u)\in[0,T]$. 
In this case, $f_\alpha^{(2)}(u)=f_{\vect g}^{(1)}(u)$ and $f_\alpha^{(2)}(v)=f_{\vect g}^{(1)}(v)$,
and the claim holds.

\item Both $f_{\vect g}^{(1)}(v),f_{\vect g}^{(1)}(u)\in[T,\frac{4}{3}T]$. 
Then $|f_\alpha^{(2)}(u)-f_\alpha^{(2)}(v)|=\alpha|f_{\vect g}^{(1)}(u)-f_{\vect g}^{(1)}(v)|$, 
and the claim holds.

\item Both $f_{\vect g}^{(1)}(v),f_{\vect g}^{(1)}(u)\in[-\tfrac13 T,0]$. 
This case is symmetric to the previous one.

\item $f_{\vect g}^{(1)}(v)\in[-\tfrac13 T,0]$ and $f_{\vect g}^{(1)}(u)\in[T,\frac{4}{3}T]$. 
Then for all $\alpha\in[0,1]$, 
we have $f_\alpha^{(2)}(u) - f_\alpha^{(2)}(v) \geq \frac{2}{3}T - \frac{1}{3}T = \tfrac13 T$, while $\abs{f_{\vect g}^{(1)}(u)-f_{\vect g}^{(1)}(v)} \leq\frac{5}{3}T$, 
and the claim follows.

\item $f_{\vect g}^{(1)}(v)\in[0,T]$ and $f_{\vect g}^{(1)}(u)\in[T,\frac{4}{3}T]$. 
Here we handle two sub-cases, depending on the size of the flipped region 
relative to $L\eqdef f_{\vect g}^{(1)}(u) - f_{\vect g}^{(1)}(v) > 0$.
\begin{itemize}
\item Assume $f_{\vect g}^{(1)}(u) -T \leq \tfrac12 L$.
Then for $\alpha=\frac13$ we have
$f_\alpha^{(2)}(u) - f_\alpha^{(2)}(v) \geq \tfrac12 L - \alpha\cdot \tfrac12 L =\tfrac13 L$.

\item Otherwise, $f_{\vect g}^{(1)}(u) -T > \tfrac12 L$.
Then the possible images of $u$ under the two different 
$\alpha\in\aset{\tfrac13,1}$ are ``far'' apart, 
namely, $f_{1/3}^{(2)}(u) - f_1^{(2)}(u) \geq \tfrac23 \cdot \tfrac12 L = \tfrac13 L$.
Hence, under a uniformly random $\alpha\in\aset{\tfrac13,1}$, the expected distance 
between the image of $u$ and any fixed point is at least $\tfrac16 L$,
and the image of $v$ is indeed fixed regardless of $\alpha$ 
to be $f_\alpha^{(2)}(v) = f_{\vect g}^{(1)}(v)$.
\end{itemize}
We see that in both sub-cases 
$\EX_\alpha \abs{ f_\alpha^{(2)}(u) - f_\alpha^{(2)}(v) } \ge \tfrac16 L$.

\item $f_{\vect g}^{(1)}(u)\in[0,T]$ and $f_{\vect g}^{(1)}(v)\in[-\frac{1}{3}T,0]$. 
This case is symmetric to the previous one.
\end{itemize}
\end{proof}

We proceed with the proof of \cref{lmm_nodenseball}. 
Applying \cref{flipping_claim2} to all $u,v\in V$, we get that
\[ 
  \EX_{\alpha\in\aset{1/3,1}} \sum_{u,v\in V} \mu(u)\mu(v)\cdot \left|f_\alpha^{(2)}(u)-f_\alpha^{(2)}(v)\right| 
  \geq \tfrac16 \sum_{u,v\in V} \mu(u)\mu(v)\cdot\left|f_{\vect g}^{(1)}(u)-f_{\vect g}^{(1)}(v)\right|,
\]
and we can fix $\alpha\in\aset{\tfrac13,1}$ that attains this inequality.
For the same value of $\alpha$, we have by \cref{flipping_claim1} that also
\[ 
  \sum_{u,v\in V} \fcap(u,v) \cdot\left|f_\alpha^{(2)}(u)-f_\alpha^{(2)}(v)\right| 
  \leq \sum_{u,v\in V} \fcap(u,v) \cdot\left|f_{\vect g}^{(1)}(u)-f_{\vect g}^{(1)}(v)\right|. \]
Putting these together with \cref{maineq}, we get
\[
  \frac
  { \sum_{u,v\in V} \fcap(u,v) \cdot\left|f_\alpha^{(2)}(u)-f_\alpha^{(2)}(v)\right|}
  {\sum_{u,v\in V} \mu(u)\mu(v) \cdot\left|f_\alpha^{(2)}(u)-f_\alpha^{(2)}(v)\right|}
  \leq 
  O(1) \cdot 
  \frac 
  { \sum_{u,v\in V} \fcap(u,v) \cdot\norm{\vect x_u - \vect x_v}_2}
  {\sum_{u,v\in V} \mu(u)\mu(v) \cdot\norm{\vect x_u - \vect x_v}_2}.
\]
We now bound the right-hand side.
For the numerator, Jensen's inequality yields
\[
   \sum_{u,v\in V} \fcap(u,v) \cdot\norm{x_u-x_v}_2
  \leq \sqrt{ \sum_{u,v\in V} \fcap(u,v) \cdot\norm{x_u-x_v}_2^2} 
  \leq \sqrt{\SDP}.
\]
For the denominator, recall our hypothesis, which can be written as
$\Pr_{u,v\sim\mu}\left[\norm{x_u-x_v}_2 > \frac12\right] \geq \frac{1}{2}$,
and implies that 
$\EX_{u,v\sim\mu} \norm{\vect x_u - \vect x_v}_2 \geq \tfrac14$.
Putting these together gives
\[
  \frac
  { \sum_{u,v\in V} \fcap(u,v) \cdot\left|f_\alpha^{(2)}(u)-f_\alpha^{(2)}(v)\right|}
  {\sum_{u,v\in V} \mu(u)\mu(v) \cdot\left|f_\alpha^{(2)}(u)-f_\alpha^{(2)}(v)\right|}
  \leq 
  O(1)\cdot\frac{\sqrt{\SDP}}{1/4} = O(\sqrt{\SDP}).
\]
Now applying \cref{lmm_l1} to $f_\alpha^{(2)}$ produces a cut of sparsity $O(\sqrt{\SDP})$. 
Moreover, $f_\alpha^{(2)}$ is confined to the interval $[0,T]$,
while $f_\alpha^{(2)}(s)=f_{\vect g}^{(1)}(s)=0$ and $f_\alpha^{(2)}(t)=f_{\vect g}^{(1)}(t)=T$, 
hence \cref{lmm_l1} ensures the cut is $st$-separating,
and this completes the proof of \cref{lmm_nodenseball}.
\end{proof}

\subsection{Proof of \protect{\cref{thm_cheeger}}}

Let $G=(V,\fcap,\mu)$ be an instance of $st$-$\cproblem{ProductSparsestCut}$ with optimum $\OPT$. Set up and solve the semi-definite program \eqref{sdp}. 
Let $\SDP$ be the optimum and $\{\vect x_v\}_{v\in V}$ a solution that attains it. 
Apply \cref{lmm_cases} to the semi-metric given by $d(u,v)=\norm{\vect x_u-\vect x_v}_2^2$. 
If the first case in \cref{lmm_cases} holds, 
use \cref{lmm_denseball} to compute a cut with sparsity $O(\SDP)$. 
Otherwise, the second case of \cref{lmm_cases} must hold, 
and then use \cref{lmm_nodenseball} to compute a cut of sparsity 
$O(\sqrt{\SDP})$. Since $\SDP\leq\OPT$, \cref{thm_cheeger} follows. 
\qed

\section{A divide-and-conquer approach for product demands}
\label{app:iterative}

We now present an algorithm for $st$-$\cproblem{ProductSparsestCut}$, 
which essentially reduces the problem to its non-$st$ version
with only a constant factor loss in the approximation ratio.
This algorithm follows the well-known divide-and-conquer approach,
carefully adapted to the requirement that $s$ and $t$ are separated,
for example, it is initialized via a minimum $st$-cut computation.
This result was obtained in collaboration with Alexandr Andoni,
and we thank him for his permission to include this material.

For simplicity, we state and prove the case of uniform demands.
The theorem immediately extends to product demands, 
i.e., reduces $st$-$\cproblem{ProductSparsestCut}$ 
to $\cproblem{ProductSparsestCut}$, 
and the same bounds on the approximation ratio $\rho(n)$ are known 
for this case.

\begin{theorem}\label{thm_uniform}
Suppose $\cproblem{UniformSparsestCut}$ admits a polynomial-time approximation
within factor $\rho(n)$.
Then $st$-$\cproblem{UniformSparsestCut}$ admits a polynomial-time approximation within factor $O(\rho(n))$.
\end{theorem}

The best approximation ratio known for $\cproblem{UniformSparsestCut}$ 
to date is $\rho=O(\sqrt{\log n})$ \cite{ARV09}.
Our result actually extends also to graphs excluding a fixed minor, 
for which the known approximation is $\rho=O(1)$ \cite{KPR93,Rao99,FT03,LS13,AGGNT13}.

\paragraph{Remark.} 
It may seem that \cref{thm_uniform} can yield also a Cheeger-type approximation for $st$-$\cproblem{ProductSparsestCut}$ 
(and thus subsume \cref{thm_cheeger}), 
by replacing the $\rho(n)$-approximation with Trevisan's 
Cheeger-type approximation algorithm for $\cproblem{ProductSparsestCut}$.
However, the analysis of \cref{thm_uniform} does not carry through;
the divide-and-conquer algorithm applies the assumed algorithm  
(for $\cproblem{ProductSparsestCut}$) to various subgraphs of the input graph,
which are all of size at most $n$,
but a Cheeger-type approximation factor on these subgraphs 
depends on their expansion after normalizing their total capacity and demand. 
Concretely, an input graph $G$ may be an expander but contain 
a small non-expanding subgraph $G'$.
A Cheeger-type approximation for $G$ should yield an $O(1)$-approximation, 
but a Cheeger-type approximation for $G'$ is super-constant
and breaks the analysis of \cref{thm_uniform}. 
Nevertheless, it remains possible that our divide-and-conquer algorithm, 
possibly with minor tweaks, does provide a Cheeger-type approximation.

\subsection{The divide-and-conquer algorithm}

Our algorithm iteratively removes a piece from the current graph,
until ``exhausting'' all the entire graph. 
During its execution, the algorithm ``records'' a list of candidate cuts, all of which are $st$-separating, and eventually returns the best cut in the list.
The idea is that our analysis can determine the ``correct'' stopping point 
using information that is not available to the algorithm, like the size of the optimum cut.
The algorithm works as follows.

\begin{algorithm}
\begin{algorithmic}[1]
\label{alg1}
\STATE compute a minimum $st$-cut $(S_0,V\setminus S_0)$ in $G$;
let ${S_0}$ be the smaller side and $s\in S_0$
\STATE \textbf{record} the cut $(S_0,V\setminus S_0)$
\label{stp:recordS0}
\STATE set $V'\gets V\setminus S_0$;\; $S\gets S_0$
\WHILE {$\card{V'} \geq 2$}
\STATE compute a $\rho$-approximate sparsest cut $(C,\bar C)$ in $G[V']$;
let $C$ be the smaller side \label{stp:rhoapp}
\label{stp:cutC}
\IF {$t\notin C$} \label{stp:cond}
\STATE set $S\gets S\cup C$ and \textbf{record} the cut $(S,V\setminus S)$
\label{stp:recordS}
\ELSE 
\STATE set $T\gets C$ and \textbf{record} the cut $(T,V\setminus T)$
\label{stp:recordT}
\ENDIF
\STATE set $V'\gets V'\setminus C$
\ENDWHILE
\RETURN a recorded cut of minimum sparsity
\end{algorithmic}
\end{algorithm}

\begin{claim}
All recorded cuts (and thus also the output cut) are $st$-separating.
\end{claim}
\begin{proof}
The cut recorded in step~\ref{stp:recordS0} is clearly $st$-separating. 
Inspecting the iterations of the main loop, 
we see they maintain that $s\in S_0\subseteq S$ and $t\notin S$,
and thus the cut recorded in step~\ref{stp:recordS} must be $st$-separating. 
Finally, when step~\ref{stp:recordT} is executed, which happens at most once,
$T=C\subseteq V\setminus S_0$ contains $t$ but not $s$, 
hence the recorded cut is $st$-separating. 
\end{proof}

\paragraph{Notation.}
Throughout the analysis, it will be convenient to work with a slightly different definition of cut sparsity, 
\begin{equation} \label{newfsp}
  \fsp(S,\bar S) = \frac{\fcap(S,\bar S)}{\min\{|S|,|\bar S|\}}.
\end{equation}
It is well-known (and easy to verify) that up to a factor of 2 and appropriate scaling, this quantity is equivalent to the one given in \cref{section_intro}.
In particular, a $\rho$-approximation under one definition 
is a $2\rho$-approximation under the other definition.

For the rest of the analysis, fix a sparsest $st$-separating cut in $G$, namely, one that minimizes \cref{newfsp}, denoted $(V_{\opt},\bar V_{\opt})$,
with $V_{\opt}$ being the smaller side, and let $\OPT=\fsp(V_{\opt},\bar V_{\opt})$.
We proceed by considering three cases, 
which correspond to the three steps (\ref{stp:recordS0}, \ref{stp:recordS}, and \ref{stp:recordT}) where the algorithm records a cut,
and can be viewed as different ``stopping points'' for the main loop.

\paragraph{Case 1: When $S_0$ is good enough.}

Suppose $\card{S_0}\geq\frac{1}{8} \card{V_{\opt}}$. 
Since the cut $(S_0,\bar S_0=V\setminus S_0)$ 
recorded in step~\ref{stp:recordS0} is a minimum $st$-cut, 
\[
  \fsp(S_0,\bar S_0) = \frac{\fcap(S_0,\bar S_0)}{|S_0|} 
  \leq \frac{\fcap(V_{\opt},\bar V_{\opt})}{|V_{\opt}|/8} 
  = 8\cdot \OPT.
\]
Thus, in this case our algorithm achieves a constant-factor approximation.

\paragraph{Further notation:}
\begin{itemize} \compactify
\item 
Denote by $(C_i,\bar C_i)$ the cut computed in iteration $i$ of step~\ref{stp:cutC}. 
Note that in this step $C_i\cup \bar C_i=V'$, rather than the entire $V$.
\item 
Let $S_i$ denote the value of $S$ at the end of iteration $i$ of the main loop.
Observe that $S_i$ is the disjoint union $S_0\cup C_1\cup C_2\cup\cdots\cup C_i$ minus the set $C_j$ containing $t$, if any. 
\item
Let $i^*\ge 0$ be the smallest such that $|S_{i^*}|\geq\frac{1}{3}|V_{\opt}|$. 
We assume henceforth that Case 1 does not hold, and thus $i^*\ge 1$.
\end{itemize}

\paragraph{Case 2: The ``standard'' case.}
We consider next what we call the standard case, 
where in the first $i^*$ iterations the condition in step~\ref{stp:cond} is met, which means that $t$ falls in the larger side of the cut $(C,\bar C)$.
In this case, $S_{i^*}=S_0\cup C_1\cup\cdots\cup C_{i^*}$.
The following two claims will be used to analyze the size and capacity of the cut produced after $i^*$ iterations.

\begin{claim}\label{clm_stdcase_dem}
$\min\{|S_{i^*}|,|V\setminus S_{i^*}|\} \geq \frac13 |S_{i^*}|$.
\end{claim}
\begin{proof}
By definition of $i^*$ we have 
$|S_{i^*-1}|<\frac{1}{3}|V_{\opt}|\leq \frac16 |V|$. 
And since $C_{i^*}$ is the smaller side of some cut, $|C_{i^*}|\leq \frac{1}{2}|V|$. 
Together, $|S_{i^*}|=|S_{i^*-1}|+|C_{i^*}| < \frac23 |V|$, 
and we get $|V\setminus S_{i^*}| > \tfrac13 |V| \geq \frac13 |S_{i^*}|$, 
as required.
\end{proof}

\begin{claim} \label{clm_stdcase_aprx}
For all $i=1,\ldots,i^*$,\; 
$\fsp(C_i,\bar C_i) \leq \tfrac32\rho\cdot \OPT$.
(Note that $(C_i,\bar C_i)$ is a cut in the induced subgraph $G[V\setminus S_{i-1}]$, whereas $(V_{\opt},\bar V_{\opt})$ is a cut in the input graph $G$.)
\end{claim}
\begin{proof}
Fix $i$ and denote by $G_i$ the induced graph at the beginning of iteration $i$,
i.e., $G_i=G[V\setminus S_{i-1}]$. 
The cut $(V_{\opt},\bar V_{\opt})$ induces in $G_i$ some cut $(U,\bar U)$, 
where $U\subseteq V_{\opt}$ and $\bar U\subseteq\bar V_{\opt}$. 
Since $i\leq i^*$, earlier iterations (before $i$) have removed from the graph 
less than $\frac13 |V_{\opt}|$ vertices, and in particular
\begin{align} \label{eq_induced_cut_dem}
  \minn{|U|,|\bar U|}
  > \minn{|V_\opt|,|\bar V_\opt|} - \tfrac{1}{3}|V_{\opt}|
  = \tfrac23 |V_{\opt}|.
\end{align}
By definition of $(U,\bar U)$ we have 
$\fcap(U,\bar U)\leq\fcap(V_{\opt},\bar V_{\opt})$, and we get 
$\fsp(U,\bar U)\leq \tfrac32 \fsp(V_{\opt},\bar V_{\opt}) = \tfrac32 \OPT$. 
The claim now follows from the fact that $(U,\bar U)$ is one possible cut in $G_i$ and the approximation guarantee used in step~\ref{stp:rhoapp}.
\end{proof}

We can now complete the proof for this standard case 
by showing that the recorded cut $(S_{i^*},\bar S_{i^*})$ is sufficiently good.
Indeed, using \cref{clm_stdcase_dem}
\begin{equation} \label{eq_stdcase1}
  \fsp(S_{i^*},V\setminus S_{i^*}) 
  \leq \frac{\fcap(S_{i^*},V\setminus S_{i^*})} {|S_{i^*}|/3} 
  \leq 3\left( \frac{\fcap(S_0,\bar S_0)}{|S_{i^*}|} +    
     \frac{\sum_{i=1}^{i^*}\fcap(C_i,\bar C_i)}{|S_{i^*}|} \right) .
\end{equation}
To bound the first summand in \cref{eq_stdcase1}, recall that 
$\fcap(S_0,\bar S_0)\leq\fcap(V_{\opt},\bar V_{\opt})$ 
and $|S_{i^*}|\geq\frac{1}{3}|V_{\opt}|$. 
To bound the second summand in \cref{eq_stdcase1}, 
we use \cref{clm_stdcase_aprx} and get
\begin{equation*} 
  \frac{\sum_{i=1}^{i^*}\fcap(C_i,\bar C_i)}{|S_{i^*}|} <
  \frac{\sum_{i=1}^{i^*}\fcap(C_i,\bar C_i)}{\sum_{i=1}^{i^*}|C_i|} \leq
  \max_{i=1,\ldots,i^*}\frac{\fcap(C_i,\bar C_i)}{|C_i|} \leq 
  \tfrac32\rho\cdot\fsp(V_{\opt},\bar V_{\opt}).
\end{equation*}
Plugging these back into \cref{eq_stdcase1} yields 
$\fsp(S_{i^*},V\setminus S_{i^*})\leq O(\rho)\cdot\OPT$,
which shows that in the standard case,
there is a recorded cut that achieves $O(\rho)$ approximation.

\paragraph{Case 3: The ``exceptional'' case}
It remains to consider the case where during the first $i^*$ iterations of the main loop, the condition in step~\ref{stp:cond} is \emph{not met} exactly once (it cannot happen more than once because the $C_i$'s are disjoint).
Let $j\leq i^*$ be the iteration in which this happens, 
and then step~\ref{stp:recordT} is executed and $T=C_j$.
Observe that $\card{S}$ is not increased in this iteration, and thus $j<i^*$.

We now break the analysis into two subcases.
The first (and simpler) subcase is when $|T|<\frac16 |V_\opt|$;
we can then think of the algorithm as if it puts $T$ ``aside''
(in step~\ref{stp:recordT}) and then 
the execution proceeds similarly to the standard case until iteration $i^*$,
at which time the cut $(S_{i^*},V\setminus S_{i^*})$ is recorded 
with $V\setminus S_{i^*}$ being in effect the union $T\cup V'$. 
We can then repeat our analysis of $\fsp(S_{i^*},V\setminus S_{i^*})$ 
from the standard case, except that \cref{eq_induced_cut_dem} is replaced with
\[
  \minn{\card{U},\card{\bar U}}
  \geq \minn{\card{V_{\opt}},\card{\bar V_{\opt}}} - \card{T} - \card{S_{i^*-1}}
  > (1- \tfrac16 - \tfrac13)\card{V_{\opt}} 
  = \tfrac12\card{V_{\opt}}.
\]
This leads again to the bound 
$\fsp(S_{i^*},V\setminus S_{i^*}) \leq O(\rho)\cdot\OPT$,
except that now the hidden constant contains another small loss.

In the second and final subcase we assume that $|T|\geq\frac16|V_\opt|$,
and show that the cut $(T,\bar T)$ recorded in step~\ref{stp:recordT}
is good enough.
Indeed, $V$ is partitioned at the end of iteration $j$ into three subsets: $S_{j-1}$, $T$, and the remaining vertices $V'=V\setminus(S_{j-1}\cup T)$.
Hence, 
\begin{equation} \label{eq_expflow0}
  \fsp(T,V\setminus T) 
  = \frac{\fcap(T,V\setminus T)} {\card{T}}
  = \frac{\fcap(T, V')}{\card{T}} + \frac{\fcap(T, S_{j-1})} {\card{T}} .
\end{equation}
Observe that \cref{clm_stdcase_aprx} can be applied to all iterations up to $j$, 
because every earlier iteration added vertices to $S$ and not to $T$.
Applying this to iteration $j$, which produces the cut $(T,V')$,
we bound the first summand in \cref{eq_expflow0} by
\begin{equation}\label{eq_expflow1}
\frac{\fcap(T, V')}{|T|} \leq O(\rho)\cdot \OPT .
\end{equation}
For the second summand in \cref{eq_expflow0}, we bound
\begin{align*}
  \fcap(T, S_{j-1}) 
  \leq \fcap(S_{j-1},V\setminus S_{j-1}) 
  &\leq \fcap(S_0,\bar S_0) + \sum_{i=1}^{j-1}\fcap(C_i,\bar C_i).
\end{align*}
Proceed now similarly to the standard case; 
recall that $\fcap(S_0,\bar S_0)\leq\fcap(V_{\opt},\bar V_{\opt})$,
and use \cref{clm_stdcase_aprx} to obtain
\[
  \sum_{i=1}^{j-1}\fcap(C_i,\bar C_i)
  \leq \sum_{i=0}^{j-1} \left( \frac23\rho\cdot\OPT\cdot|C_i| \right)
  < \frac23\rho\cdot\OPT\cdot |S_{j-1}|
  < \frac29\rho\cdot\OPT\cdot |V_\opt|.
\]
Gathering the above inequalities, we obtain
\begin{equation}\label{eq_expflow2}
  \frac{\fcap(T, S_{j-1})}{|T|} 
  \leq \frac{\fcap(V_{\opt},\bar V_{\opt}) + \frac29\rho\cdot\OPT\cdot |V_\opt| } {|V_\opt|/6} 
  \leq O(\rho)\cdot \OPT.
\end{equation}
Plugging \cref{eq_expflow1,eq_expflow2} into \cref{eq_expflow0},
we have $\fsp(T,\bar T) \leq O(\rho)\cdot \OPT$,
which shows that in this final subcase,
the cut recorded in iteration $j$ achieves $O(\rho)$ approximation.
This completes the proof of \cref{thm_uniform}.

\section{Concluding remarks}

The discrete Cheeger inequality \cite{AM85, JS88, Mihail89} can be used to approximate the conductance of a graph $G$ based on an eigenvector computation. 
Specifically, letting $\hat L_G$ denote the normalized Laplacian of $G$, 
the eigenvector associated with the second-smallest eigenvalue 
of $\hat L_G$ is the minimizer of
\begin{equation}\label{eq:cheeger}
  \min \left\{ 
    \frac{v^T \hat L_Gv}{v^Tv}:\ 
    v\neq 0,\ v \perp \mathbf1 
  \right\},
\end{equation}
where $\mathbf1$ is the all-ones vector.
The solution $v$ can be ``rounded'' into a cut of near-optimal conductance
by using a simple sweep-line procedure on the entries of $v$,
see \cite{Chung97,Spielman12} for recent presentations.
Moreover, this computation can be carried out (within reasonable accuracy) 
in near-linear time, which makes it useful in practical settings.

It is natural ask whether this approach extends to the $st$-separating setting. 
The optimization problem analogous to \cref{eq:cheeger} would have 
an additional constraint to ensure $st$-separation,
\begin{equation}\label{eq:st-cheeger}
  \min \left\{ 
    \frac{v^T \hat L_Gv}{v^Tv}:\ 
    v\neq 0,\ v \perp \mathbf1,\ \forall {i\in V},\ v_s\leq v_i\leq v_t 
  \right\}.
\end{equation}
It is not difficult to verify a solution $v$ to \cref{eq:st-cheeger} 
can be ``rounded'' to a cut achieving a Cheeger-type approximation for the $st$-conductance of $G$.
However, we currently do not know whether \cref{eq:st-cheeger} can be solved, or even approximated within constant factor, in polynomial time.

\subsubsection*{Acknowledgments}

We thank Alexandr Andoni, Aleksander Madry and Luca Trevisan 
for useful discussions at various stages of the research.
We are especially grateful to Alexandr Andoni for his permission to include the material presented in \cref{thm_uniform}.

{\small
\bibliographystyle{alphaurlinit}
\bibliography{stsparsbib,robi,stspars_application}
}

\appendix

\section{Deferred Proofs from \Cref{sec:basic}}
\label{sec:basic_proofs}

\begin{proof}[Proof of \cref{lmm_l1}]
First suppose $m=1$. Denote $f^{\min}=\min_{v\in V}f(v)$ and $f^{\max}=\max_{v\in V}f(v)$. Sample a threshold $\tau\in(f^{\min},f^{\max})$ uniformly at random, and let $S_\tau=\{v\in V:f(v)\leq\tau\}$. Note that $S_\tau\neq\emptyset,V$. Let $\chi_\tau$ denote the characteristic function of $S_\tau$. For every $u,v\in V$ we have
\[ \EX_\tau\left|\chi_\tau(u)-\chi_\tau(v)\right| = \tfrac{1}{f^{\max}-f^{\min}}\left|f(u)-f(v)\right| , \]
and hence
\[ \frac{\EX_\tau\left[ \sum_{u,v\in V}\fcap(u,v)\left|\chi_\tau(u)-\chi_\tau(v)\right| \right]}{\EX_\tau\left[ \sum_{u,v\in V}\fdem(u,v)\left|\chi_\tau(u)-\chi_\tau(v)\right| \right]} 
= \frac{\sum_{u,v\in V}\fcap(u,v)\left|f(u)-f(v)\right|}{\sum_{u,v\in V}\fdem(u,v)\left|f(u)-f(v)\right|}.
\]
Consequently, there is a choice of $\tau$ for which
\[ \frac{\sum_{u,v\in V}\fcap(u,v)\left|\chi_\tau(u)-\chi_\tau(v)\right|}{\sum_{u,v\in V}\fdem(u,v)\left|\chi_\tau(u)-\chi_\tau(v)\right|} 
\leq \frac{\sum_{u,v\in V}\fcap(u,v)\left|f(u)-f(v)\right|}{\sum_{u,v\in V}\fdem(u,v)\left|f(u)-f(v)\right|}.
\]
The left-hand side is $sp_G(S_\tau,\bar S_\tau)$, so it is a cut as needed. Observe that $f$ induces an ordering of the vertices, $f(v_1)\leq f(v_2)\leq\ldots\leq f(v_n)$, and $S_\tau$ is a prefix of the vertices by that ordering. Hence, it can be found efficiently by enumerating over all prefixes, 
as there are less than $n$ of them. 
Finally, if $f$ is $st$-sandwiching then $f(s)=f^{\min}$ and $f(t)=f^{\max}$, which necessarily implies $s\in S_\tau$ and $t\in\bar S_\tau$, and $(S_\tau,\bar S_\tau)$ is an $st$-separating cut.

This proves the lemma for the $m=1$ case. To remove this assumption,  denote $f=(f_1,\ldots,f_m)$ and observe that
\begin{align*}
  \frac{\sum_{u,v\in V}\fcap(u,v)\norm{f(u)-f(v)}_1}{\sum_{u,v\in V}\fdem(u,v)\norm{f(u)-f(v)}_1}
  & =
  \frac{\sum_{k=1}^m\left(\sum_{u,v\in V}\fcap(u,v)\left|f_k(u)-f_k(v)\right|\right)}{\sum_{k=1}^m\left(\sum_{u,v\in V}\fdem(u,v)\left|f_k(u)-f_k(v)\right|\right)} \\
  &\geq \min_{k=1,\ldots,m}\frac{\sum_{u,v\in V}\fcap(u,v)\left|f_k(u)-f_k(v)\right|}{\sum_{u,v\in V}\fdem(u,v)\left|f_k(u)-f_k(v)\right|},
\end{align*}
so we can find an optimal coordinate $f_k$ of $f$ (one achieving the minimum)
and apply to it the above argument for dimension $m=1$.
\end{proof}

\begin{proof}[Proof of \cref{prop:stsep}]
Fix $u,v\in V$.
By the triangle inequality, $d(u,v)\leq d(u,s)+d(s,v)$ 
and also $d(u,v)\leq d(u,t)+d(t,v)$. 
Sum these inequalities and apply the $st$-separation property, to get
\[ 
  2d(u,v)
  \leq \left[ d(u,s)+d(u,t)\right ]+\left[ d(s,v)+d(t,v) \right]
  = 2d(s,t). 
\qedhere
\]
\end{proof} 

\begin{proof}[Proof of \cref{prop_lip}]
Let $\sigma\in\{\pm 1\}$. By the triangle inequality, $d(u,v)\geq\left|d(u,s)-d(v,s)\right|$,
and similarly by \cref{lmm_triangle}, $d(u,v)\geq\left|d(u,A)-d(v,A)\right|$. 
Using these,
\begin{align*}
  \left|f_A^\sigma(u)-f_A^\sigma(v)\right| 
  &= \tfrac{1}{2}\left|\left[d(u,s)-d(v,s)\right]+\sigma\left[d(u,A)-d(v,A)\right]\right| \\
  &\leq \tfrac{1}{2}\left|d(u,s)-d(v,s)\right| + \tfrac{1}{2}\left|d(u,A)-d(v,A)\right| \\
  &\leq \tfrac{1}{2}d(u,v) + \tfrac{1}{2}d(u,v) = d(u,v) .
\qedhere
\end{align*}
\end{proof}

\begin{proof}[Proof of \cref{prop_goodsign}]
Denote $x=\frac{1}{2}\left[d(u,A)-d(v,A)\right]$ and $y=\frac{1}{2}\left[d(u,s)-d(v,s)\right]$. Then,
\[ 
  \norm{f_A^\pm(u)-f_A^\pm(v)}_1 
  = \left|f_A^+(u)-f_A^+(v)\right| + \left|f_A^-(u)-f_A^-(v)\right|
  = \left|y+x\right| +\left|y-x\right|
  \geq \left|x\right|, 
\]
as needed, where the inequality is since either $|y+x|\geq|x|$ or $|y-x|\geq|x|$, depending on whether $x,y$ have the same or opposite signs.
\end{proof}

\begin{proof}[Proof of \cref{prop_separation}]
For the $f^+_A$ coordinate,
\begin{itemize} \compactify
\item By \cref{lmm_triangle}, $f_A^+(s)=\frac{1}{2}d(s,A)\leq\frac{1}{2}\left(d(v,s)+d(v,A)\right)=f_A^+(v)$.
\item By \cref{lmm_triangle}, $d(v,A)\leq d(v,t)+d(t,A)$. By the $st$-separation, $d(v,t)=d(s,t)-d(v,s)$. Plugging and rearranging we get $d(v,s)+d(v,A)\leq d(s,t)+d(t,A)$, so $f_A^+(v)\leq f_A^+(t)$.
\end{itemize}
For the $f^-_A$ coordinate,
\begin{itemize} \compactify
\item By \cref{lmm_triangle}, $d(v,A)\leq d(v,s)+d(s,A)$, and hence $f_A^-(s)=-\frac{1}{2}d(s,A)\leq\frac{1}{2}\left(d(v,s)-d(v,A)\right)=f_A^-(v)$.
\item By the $st$-separation, $f_A^-(t)=\frac{1}{2}\left(d(s,t)-d(t,A)\right)=\frac{1}{2}\left(d(s,v)+d(v,t)-d(t,A)\right)$. By \cref{lmm_triangle}, $d(v,t)-d(t,A)\geq-d(v,A)$. Combining these yields $f_A^-(t)\geq\frac{1}{2}\left(d(s,v)-d(v,A)\right)=f_A^-(v)$.
\end{itemize}
\end{proof}

\end{document}